\documentclass[11pt]{article}
\usepackage[a4paper, nohead, scale={0.7,0.7}]{geometry}
\usepackage{amsmath,amssymb,amsfonts,latexsym, amsthm,array}
\usepackage{tikz}

\usepackage[boxruled,vlined]{algorithm2e}
\newtheorem{defi}{Definition}[section]
\newtheorem{theo}[defi]{Theorem}
\newtheorem{lem}[defi]{Lemma}
\newtheorem{rem}[defi]{Remark}

\usepackage[utf8]{inputenc}
\newcommand{\SP}{\ensuremath{{\mathcal{SP}}}}

\newcommand{\pin}[2]{
\fill (#1,#2) circle (.2);
}
\newcommand{\pinU}[2]{
\fill (#1,#2) circle (.2);
\draw [thick] (#1,#2) -- (#1,#2-.45); 
}

\newcommand{\nameUnder}[3]{
\draw (#1,#2-1) node {#3};
}
\newcommand{\pinL}[2]{
\fill (#1,#2) circle (.2);
\draw [thick] (#1,#2) -- (#1+.45,#2); 
}
\newcommand{\pinR}[2]{
\fill (#1,#2) circle (.2);
\draw [thick] (#1,#2) -- (#1-.45,#2); 
}

\newcounter{indice}

\newcommand{\permutation}[1]{
\setcounter{indice}{0};
\foreach \i in {#1} 
\addtocounter{indice}{1};

\addtocounter{indice}{1}
\draw [help lines] (1,1) grid (\theindice,\theindice);

\setcounter{indice}{1};

\foreach \i in { #1 } {
\draw (\theindice+.5,\i+.5) [fill] circle (.2);
\addtocounter{indice}{1};
}
\addtocounter{indice}{-1};

}

\begin{document}

\title{Deciding the finiteness of the number of simple permutations contained in a
  wreath-closed class is polynomial\footnote{This work was completed
    with the support of the ANR project GAMMA number 07-2\_195422}.}
\author{Fr\'ed\'erique Bassino\\
LIPN UMR 7030, Universit\'e
   Paris 13 and CNRS, \\ 99, avenue J.- B. Cl\'ement, 93430 Villetaneuse, France.\\
  \and  Mathilde Bouvel\\
LaBRI UMR 5800, Universit\'e de Bordeaux and
 CNRS,\\ 351, cours de la Libération, 33405 Talence cedex, France.\\
\and  Adeline  Pierrot\\
LIAFA UMR 7089, Universit\'e Paris Diderot and
 CNRS,\\ Case 7014, 75205 Paris cedex 13, France.\\
\and Dominique Rossin\\
LIX UMR 7161, Ecole Polytechnique and CNRS, \\91128  Palaiseau, France.
}

\date{~} 

\maketitle

\begin{abstract}
We present an algorithm running in time ${\mathcal O}(n \log n)$ which
decides if a wreath-closed permutation class $Av(B)$ given by its
finite basis $B$ contains a finite number of simple permutations. The
method we use is based on an article of Brignall, Ru{\v{s}}kuc and
Vatter \cite{BRV06} 
which presents a decision procedure (of high complexity) for solving
this question, without the assumption that $Av(B)$ is
wreath-closed. Using combinatorial, algorithmic and language theoretic
arguments together with one of our previous results on
pin-permutations \cite{BBR09}, we are able to transform the problem into a
co-finiteness problem in a complete deterministic automaton.
\end{abstract}

\section{Introduction}

Permutation classes were first introduced in the literature by Knuth
in \cite{Knuth:ArtComputerProgramming:1:1973}, where the class of
permutations sortable through one stack is characterized as the
permutations avoiding the pattern $231$. This result has been the
starting point of the study of permutation classes and
pattern-avoiding permutations in combinatorics. The study of
permutation classes has been mostly interested in enumeration
questions as testified by the survey \cite{KiMa03} and its
references. The predominance of the counting questions certainly finds
an explanation in the Stanley-Wilf conjecture, stating that the
enumeration sequence $(s_n)_n$ of any (non trivial) permutation class
is at most simply exponential in the length $n$ of the permutations
(as opposed to $n!$ in general). This conjecture has been proved by
Marcus and Tardos in 2004 \cite{MaTa04}, and this result can be
considered as one of the first general results on permutation classes,
that is to say a result that deals with \emph{all} permutation
classes. More recently, some other general results dealing with wide
families of permutation classes have been described
\cite{AA05,ALR05,ARS09,BHV06a,BRV06,Vat05}. In particular, Albert and
Atkinson \cite{AA05} proved some sufficient conditions for the
generating function $S(x) = \sum s_n x^n$ of a class to be
algebraic. This is also the direction chosen in this article, where we
are interested in describing an efficient algorithm to decide the
finiteness of the number of simple permutations in a class, for
\emph{any} wreath-closed permutation class.

To be more precise, in a series of three articles
\cite{BHV06a,BHV06b,BRV06} Brignall {\em et al.}  prove that it is
decidable to know if a permutation class of finite basis contains a
finite number of simple permutations, which is a sufficient condition
for the generating function to be algebraic.  Every algorithm involved
in this decision procedure is polynomial except the algorithm deciding
if the class contains arbitrarily long proper pin-permutations.

In \cite{BBR09} a detailed study of pin-permutations is performed.  We
use some of the properties of the simple pin-permutations established
in \cite{BBR09} to give a polynomial-time algorithm for the preceding
question in the restricted case of wreath-closed permutation classes,
that is to say the classes of permutations whose bases contain only
simple permutations.  More precisely, we give a ${\mathcal O}(n \log
n)$ algorithm to decide if a finitely based wreath-closed class of
permutations $Av(\pi^{(1)},\ldots,\pi^{(k)})$ contains a finite number
of simple permutations where $n = \sum |\pi^{(i)}|$. A key ingredient
of this procedure is the transformation of a containment relation
involving permutations into a factor relation between words. As a
consequence deciding the finiteness of the number of proper
pin-permutations is changed into testing the co-finiteness of a
regular language given by a complete deterministic automaton.

The paper is organized as follows. We first recall basic definitions
and known results that will be used in the sequel. In
Section~\ref{sec:containment} we establish, in the special case of
simple patterns and proper pin-permutations, some links between pattern
containment relation on permutations and factor relation
between words. Finally Section~\ref{sec:complexity} is devoted to the
presentation of a polynomial algorithm deciding the finiteness of the
number of proper pin-permutations contained in a wreath-closed
permutation class.
  
\section{Background}\label{sec:preliminary}

\subsection{Definitions}
We recall in this section a few definitions about permutations, pin
representations and pin words. More details can be found in
\cite{BHV06a,BRV06, BBR09}. A permutation $\sigma \in S_n$ is a
bijective function from $\{1,\ldots ,n\}$ onto $\{1,\ldots ,n\}$. We
either represent a permutation by a word $\sigma=2\,3\,1\,4$ or its
{\it diagram} (see Figure \ref{fig:definitions}). A permutation $\pi =
\pi_1 \pi_2 \ldots \pi_k$ is a {\it pattern} of a permutation $\sigma
= \sigma_1 \sigma_2 \ldots \sigma_n$, and we write $\pi \leq \sigma$
if and only if there exist $1 \leq i_1 < i_2 < \ldots < i_k \leq n$
such that $\sigma_{i_1}\ldots \sigma_{i_k}$ is order isomorphic to
$\pi$.  We also say that $\sigma$ \emph{involves} or \emph{contains}
$\pi$.  If $\pi$ is not a pattern of $\sigma$ we say that $\sigma$
{\it avoids} $\pi$. A permutation class $Av(B)$ -- where $B$ is a
finite or infinite antichain of permutations called the {\it basis} --
is the set of all permutations avoiding every element of $B$. A
permutation is called {\it simple} if it contains no block,
\emph{i.e.} no mapping from $\{i,\ldots ,(i+l)\}$ to $\{j,\ldots
,(j+l)\}$, except the trivial ones corresponding to $l=0$ or $i=j=1$
and $l=n-1$.  Wreath-closed permutation classes have been introduced
in \cite{AA05} in terms of substitution- or wreath-product of
permutations. This original definition is not crucial to our work, and
we prefer to define them by the characterization proved in
\cite{AA05}: a permutation class $Av(B)$ is said to be {\em
  wreath-closed} when its basis $B$ contains only simple permutations.

In the following we study wreath-closed classes with finite basis.
Note that it is not a restriction for our purpose: from \cite{AA05}, when the
basis is infinite we know that the number of simple permutations in the class is
infinite.  Our goal is indeed to check whether a wreath-closed class contains a finite
number of simple permutations, ensuring in this way that its
generating function is algebraic \cite{AA05}. As we shall see in the
following, a class of particular permutations, called the
pin-permutations, plays a central role in the decision procedure of
this problem. For this reason, we record basic definitions and results
related with these pin-permutations.

A pin in the plane is a point at integer coordinates. A pin $p$ {\em
  separates} - horizontally or vertically - the set of pins $P$ from
the set of pins $Q$ if and only if a horizontal - resp. vertical -
line drawn across $p$ separates the plane into two parts, one of which
contains $P$ and the other one contains $Q$. A pin sequence is a
sequence $(p_1,\ldots,p_k)$ of pins in the plane such that no two
points lie in the same column or row and for all $i \geq 2$, $p_i$
lies outside the bounding box of $\{p_1,\ldots ,p_{i-1}\}$ and
respects one of the following conditions:
\begin{itemize}
\item $p_i$ separates $p_{i-1}$ from $\{p_1,\ldots,p_{i-2}\}$.
\item $p_i$ is independent from $\{p_1,\ldots,p_{i-1}\}$, {\em i.e.},
  it does not separate this set into two non empty sets.
\end{itemize}
A pin sequence represents a permutation $\sigma$ if and only if it is
order isomorphic to its diagram. We say that a permutation $\sigma$ is
a \emph{pin-permutation} if it can be represented by a pin sequence,
which is then called a \emph{pin representation} of $\sigma$. Not all
permutations are pin-permutations (see for example the permutation
$\sigma$ of Figure \ref{fig:definitions}).
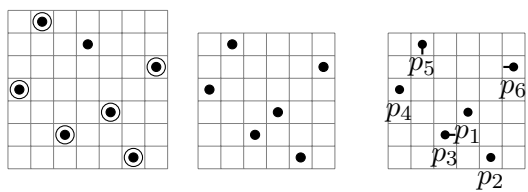
\begin{figure}[htbp]
  \begin{center}
    \begin{tikzpicture}
      \begin{scope}[scale=.3]
        \permutation{4,7,2,6,3,1,5} \draw (1.5,4.5) circle (.4 cm);
        \draw (2.5,7.5) circle (.4 cm); \draw (3.5,2.5) circle (.4
        cm); \draw (5.5,3.5) circle (.4 cm); \draw (6.5,1.5) circle
        (.4 cm); \draw (7.5,5.5) circle (.4 cm);
      \end{scope}
      \begin{scope}[xshift=2.5cm,scale=.3]
        \permutation{4,6,2,3,1,5}
      \end{scope}
      \begin{scope}[xshift=5cm,scale=.3]
        \draw (1,1) [help lines] grid +(6,6); \pin{4.5}{3.5}
        \nameUnder{4.5}{3.5}{$p_1$} \pin{5.5}{1.5}
        \nameUnder{5.5}{1.5}{$p_2$} \pinL{3.5}{2.5}
        \nameUnder{3.5}{2.5}{$p_3$} \pin{1.5}{4.5}
        \nameUnder{1.5}{4.5}{$p_4$} \pinU{2.5}{6.5}
        \nameUnder{2.5}{6.5}{$p_5$} \pinR{6.5}{5.5}
        \nameUnder{6.5}{5.5}{$p_6$}
      \end{scope}
    \end{tikzpicture}
    \caption{The permutation $\sigma=4\,7\,2\,6\,3\,1\,5$, the pattern
      $\pi=4\,6\,2\,3\,1\,5$ and a pin representation of
      $\pi$. $14L2UR$ (if we place $p_{0}$ between $p_{3}$ and
      $p_{1}$) and $3DL2UR$ are pin words corresponding to this pin
      representation.}
    \label{fig:definitions}
  \end{center}
\end{figure}

A {\it proper} pin representation is a pin representation in which
every pin $p_i$, for $i\geq 3$, separates $p_{i-1}$ from $\{p_1,
\ldots, p_{i-2}\}$. A {\it proper} pin-permutation is a permutation
that admits a proper pin representation.

\begin{rem}\label{rem:simplepin}
  A pin representation of a simple pin-permutation is always proper as
  any independent pin $p_i$ with $i\geq 3$ creates a block
  corresponding to $\{p_1, \ldots, p_{i-1}\}$.
\end{rem}

Pin representations can be encoded by words on the alphabet
$\{1,2,3,4,U,D,L,R\}$ called {\it pin words}. Consider a pin
representation $(p_1,\ldots,p_n)$ and choose an arbitrary origin $p_0$
in the plane such that it extends the pin representation to a pin
sequence $(p_0,p_1,\ldots,p_n)$. Then every pin $p_1,\ldots,p_n$ is
encoded by a letter according to the following rules:
\begin{itemize}
\item The letter associated to $p_i$ is $U$ -resp. $D,L,R$- if and
  only if $p_i$ separates $p_{i-1}$ and $\{p_0,p_1,\ldots,p_{i-2}\}$
  from the top -resp. bottom, left, right-.
\item The letter associated to $p_i$ is $1$ -resp. $2,3,4$- if and
  only if $p_i$ is independent from $\{p_0,p_1,\ldots,p_{i-1}\}$ and
  is situated in the up-right -resp. up-left, bottom-left,
  bottom-right- corner of the bounding box of
  $\{p_0,p_1,\ldots,p_{i-1}\}$.
\end{itemize}
This encoding is summarized by Figure \ref{fig:quadrant}.  The region
encoded by $1$ is called the first {\em quadrant}. The same goes for
$2,3,4$.  The letters $L,R,U,D$ are called {\em directions}, while
$1,2,3$ and $4$ are {\em numerals}. An important remark is that the
definition of pin words implies that they do not contain any of the
factors $UU, UD, DU, DD, LL, LR, RL$ and $RR$.
\begin{figure}[htbp]
  \begin{minipage}[t]{.5\linewidth}
    \begin{center}
      \begin{tikzpicture}[scale=.6]
        \useasboundingbox (0,-1) (3,3); \draw[help lines] (0,0) grid
        +(3,3); \fill (1,1) rectangle +(1,1); \draw (0.5,0.5) node
        {3}; \draw (1.5,0.5) node {D}; \draw (2.5,0.5) node {4}; \draw
        (2.5,1.5) node {R}; \draw (2.5,2.5) node {1}; \draw (1.5,2.5)
        node {U}; \draw (0.5,2.5) node {2}; \draw (0.5,1.5) node {L};
      \end{tikzpicture}
      \caption{Encoding of pins by letters.}
      \label{fig:quadrant}
    \end{center}
  \end{minipage}
  \begin{minipage}[t]{.5\linewidth}
    \begin{center}
      \begin{tikzpicture}[scale=.37]
        \draw (2,2) [fill] circle (.2); \draw
        (4,4) [fill] circle (.2); \draw (1.5,1.5) node
        {\begin{small}$p_1$\end{small}}; \draw (4.5,4.5)
        node {\begin{small}$p_2$\end{small}}; \draw (1,1) node
        {\begin{small}$11$\end{small}}; \draw (1,3) node
        {\begin{small}$41$\end{small}}; \draw (1,5) node
        {\begin{small}$4R$\end{small}}; \draw (3,1) node
        {\begin{small}$21$\end{small}}; \draw (3,3) node
        {\begin{small}$31$\end{small}}; \draw (3,5) node
        {\begin{small}$3R$\end{small}}; \draw (5,1) node
        {\begin{small}$2U$\end{small}}; \draw (5,3) node
        {\begin{small}$3U$\end{small}}; \draw[thick] (0,2) -- (6,2);
        \draw[thick] (0,4) -- (6,4); \draw[thick] (2,0) -- (2,6);
        \draw[thick] (4,0) -- (4,6);
      \end{tikzpicture}
      \caption{The two letters in each cell indicate the first two
        letters of the pin word encoding $(p_1, \ldots, p_n)$ when $p_0$ is
        taken in this cell.}
      \label{fig:origine}
    \end{center}
  \end{minipage}
\end{figure}

To each pin word corresponds a unique pin representation, hence a
unique permutation but each pin-permutation of length greater than $1$
has at least $6$ pin words associated to it.  The reason is that for
any pin representation, there are $8$ possible placements of $p_0$
w.r.t. $p_1$ and $p_2$, among which at least $6$ give a possible
prefix of a pin word (see Figure \ref{fig:origine} for an
example). On Figure \ref{fig:origine}, the two prefixes $4R$ and $3R$ (resp.. $2U$ and $3U$) may be
excluded, when $p_3$ is encoded by $R$ or $L$ (resp. $U$ or $D$).

A {\em strict} (resp. {\em quasi-strict}) pin word is a pin word of
length at least $2$ that begins by a numeral (resp. two numerals)
followed only by directions.

\begin{rem} \label{rem:proper_strict}
Encodings of proper pin-permutations~
\begin{itemize}
\item[a.]  Strict and quasi-strict pin
  words are the encodings of proper pin representations.
\item[b.]  However a pin-permutation is proper if and only if it
  admits a strict pin word among its encodings.
\end{itemize}
\end{rem}

The language $\SP$ of strict pin words can be described by the
following regular expression:
$$(1+2+3+4)
\Big( (\epsilon+ L+ R) (U+D) \big((L+R)(U+D)\big)^{*} + (\epsilon+ U+
D) (L+ R) \big((U+D)(L+R)\big)^{*} \Big).$$

\subsection{Some known results}

In \cite{BRV06} Brignall {\it et al.} studied conditions for a class
to contain an infinite number of simple permutations. Introducing
three new kinds of permutations they show that this problem is
equivalent to looking for an infinite number of permutations of one of
these three simpler kinds.

\begin{theo}\cite{BRV06}\label{thm:brignall}
 A permutation class $Av(B)$ contains an infinite number of simple
 permutations if and only if it contains either:
\begin{itemize}
 \item An infinite number of wedge simple permutations.
\item An infinite number of parallel alternations.
\item An infinite number of proper pin-permutations.
\end{itemize}
\end{theo}

The definitions of the wedge simple permutations and the parallel
alternations are not crucial to our work, hence we refer the reader to
\cite{BRV06} for more details. What is however important for our
purpose is to be able to test whether a class given by its finite
basis contains an infinite number of permutations of these
kinds. Alternations and wedge simple permutations are well
characterized in \cite{BRV06}, where it is shown that it is easy to
deal with this problem using the three following lemmas.

\begin{lem}\cite{BRV06}\label{lem:alternation}
The permutation class $Av(B)$ contains only finitely many parallel
alternations if and only if its basis $B$ contains an element of every
symmetry of the class $Av(123, 2413, 3412)$.
\end{lem}

\begin{lem}\cite{BRV06}\label{lem:wedge1}
The permutation class $Av(B)$ contains only finitely many wedge simple
permutations of type 1 if and only if $B$ contains an element of every
symmetry of the class $Av(1243, 1324, 1423, 1432, 2431, 3124, 4123,
4132, 4231, 4312)$.
\end{lem}

\begin{lem}\cite{BRV06}\label{lem:wedge2}
 The permutation class $Av(B)$ contains only finitely many wedge simple
 permutations of type 2 if and only if $B$ contains an element of every
 symmetry of the class $Av(2134, 2143,$ $ 3124, 3142, 3241, 3412, 4123,
 4132, 4231, 4312)$.
\end{lem}

With these lemmas, it is possible to decide in polynomial time whether a class contains a finite number of wedge simple permutations or of parallel alternations. More precisely, we have:

\begin{lem}
\label{lem:complexity}
Testing whether a finitely based class $Av(B)$ contains finitely many parallel alternations (resp. wedge simple permutations of type $1$, resp. wedge simple permutations of type $2$) can be done in $\mathcal{O}(n \log n)$ time, where $n = \sum_{\pi \in B} |\pi|$.
\end{lem}

\begin{proof}
By Lemmas \ref{lem:alternation}, \ref{lem:wedge1} and \ref{lem:wedge2}, deciding if a class $Av(B)$ contains a finite number of wedge simple permutations or parallel alternations is equivalent to checking if there exists an element of $B$ in every symmetric class of special pattern avoiding permutation classes, where the bases are composed only of permutations of length at most $4$. From \cite{AAAH01} checking whether a permutation $\pi$ avoids some patterns of length at most $4$ can be done in ${\mathcal O}(|\pi| \log |\pi|)$. This leads to a ${\mathcal O}(n \log n)$ algorithm for deciding whether the numbers of parallel alternation and of wedge simple permutations in the class are finite.
\end{proof}

In \cite{BRV06} Brignall {\it et al.} also proved that it is
decidable to know if a class contains a infinite number of proper
pin-permutations using language theoretic arguments.  Analyzing their
procedure, we can prove that it has an exponential complexity due to
the resolution of a co-finiteness problem for a regular language given
by a non-deterministic automaton.
As said before our goal in this paper is to solve this same problem
in polynomial time for the wreath-closed classes.

\section{Pattern containment and pin words} \label{sec:containment}

In this section we show how to transform into a factor relation between words the pattern containment
relation of a simple permutation pattern in a proper pin-permutation. More precisely, let $Av(B)$ be a
finitely based wreath-closed class of permutations, that is to say
such that its basis $B$ is finite and contains only simple
permutations. We prove that the set of strict pin words corresponding
to permutations that contain an element of $B$ is characterized as the
set of all strict pin words whose images by a particular bijection
(denoted by $\phi$ in the sequel) contain some factors.

First recall the definition of the partial order $\preceq $ on pin
words introduced in \cite{BRV06}.  

\begin{defi}
  Let $u$ and $w$ be two pin words.  We decompose $u$ in terms of its
  strong numeral-led factors as $u = u^{(1)} \ldots u^{(j)}$, {\em a
    strong numeral-led factor} being a sequence of contiguous letters
  beginning with a numeral and followed by any number of directions
  (but no numerals).  We then write $u \preceq w$ if $w$ can be
  chopped into a sequence of factors $w=v^{(1)}w^{(1)} \ldots
  v^{(j)}w^{(j)}v^{(j+1)}$ such that for all $i \in \{1,\ldots, j\}$:
\begin{itemize}
\item if $w^{(i)}$ begins with a numeral then $w^{(i)} = u^{(i)}$, and
\item if $w^{(i)}$ begins with a direction, then $v^{(i)}$ is
nonempty, the first letter of $w^{(i)}$ corresponds to a point lying
in the quadrant specified by the first letter of $u^{(i)}$, and all
other letters in $u^{(i)}$ and $w^{(i)}$ agree.
\end{itemize}
\label{def:preceq}
\end{defi}

This order is closely related to the pattern containment order $\leq$ on
permutations.
\begin{lem}\cite{BRV06}\label{csq ordre}
If the pin word $w$ corresponds to the permutation $\sigma$ and $\pi
\leq \sigma$ then there is a pin word $u$ corresponding to $\pi$ with
$u \preceq w$.  Conversely if $u \preceq w$ then the permutation
corresponding to $u$ is contained in the permutation corresponding to
$w$.
\end{lem}

In what follows, $\sigma$ is a proper pin-permutation. So we can
choose a strict pin word $w$ that encodes $\sigma$ (see Remark
\ref{rem:proper_strict} b.). As a consequence of Lemma \ref{csq ordre},
checking whether a permutation $\pi$ is a pattern of $\sigma$ is
equivalent to checking whether there exists a pin word $u$
corresponding to $\pi$ with $u \preceq w$. Additionally, we show that
when $\pi$ is simple, we can associate to each strict
(resp. quasi-strict) pin word $v= v_1 v_2 \ldots v_n$ of $\pi$ a word
$\phi(v)$ (resp.  $\phi(v_2 \ldots v_n)$) that does not contain
numerals and such that the pattern involvement problem is equivalent
to checking if $\phi(w)$ has a factor of the form $\phi(v)$ for $v$
strict or $\phi(v_2 \ldots v_n)$ for $v$ quasi-strict encoding $\pi$.

\begin{defi}
Let $M$ be the set of words of length greater than or equal to $3$
over the alphabet ${L,R, U, D}$ such that ${R,L}$ is followed by
${U,D}$ and conversely. 

We define a bijection $\phi$ from $\SP$ to $M$ as follows. For any
strict pin word $u \in \SP$ such that $u=u' u''$ with $|u'|=2$,
we set $\phi (u) = \varphi(u') u''$ where $\varphi$ is given by:
\begin{center}
\begin{tabular}{|c||c||c||c|}
\hline 
$1R\mapsto RUR$ & $2R\mapsto LUR$   & $3R\mapsto LDR$   & $4R\mapsto RDR$ \\ 
 $1L\mapsto RUL$ & $2L\mapsto LUL$ & $3L\mapsto LDL$  & $4L\mapsto RDL$\\
 $1U\mapsto URU$ & $2U\mapsto ULU$ & $3U\mapsto DLU$ &  $4U\mapsto DRU$ \\
$1D\mapsto URD$ & $2D\mapsto ULD$ & $3D\mapsto DLD$ & $4D\mapsto DRD$ \\
\hline
\end{tabular}
\end{center}
\end{defi}

For any $n \geq 2$, the map $\phi$ is a bijection from the set
$\SP_{n}$ of strict pin words of length $n$ to the set $M_{n+1}$ of words
of $M$ of length $n+1$. 
Furthermore, it satisfies, for any $u \in \SP$, $u_{i}=\phi (u)_{i+1}$ for any $i\geq 2$.

In the above table, we can notice that, for any $u \in \SP$, the first
two letters of $\phi (u)$ are sufficient to determine the first letter
of $u$ (which is a numeral). Thus it is natural to extend the
definition of $\phi$ to words of length $1$ (that do not belong to
\SP~by definition) by setting $ \phi(1) = \{UR,RU\}, \phi(2) =
\{UL,LU\}, \phi(3) = \{DL,LD\}$ and $\phi(4) = \{RD,DR\}$, and by
defining consistently $\phi^{-1}(v) \in \{1,2,3,4\}$ for any $v$ in
$\{LU,LD,RU,RD,UL,UR,DL,DR\}$.

Notice that our bijection consists of replacing the only numeral in
any strict pin word by two directions.  Lemma \ref{lem:quadrant} below
shows that for each strict pin word $w$, we know in which quadrant
lies every pin of the pin representation corresponding to $w$.

\begin{lem}\label{lem:quadrant}
  Let $w$ be a strict pin word and $p$ the pin representation
  corresponding to $w$. For any $i\geq 2$, set
$$q(w_{i-1},w_i) = \begin{cases} \phi^{-1}(w_{i-1}w_i) \textrm{ if }
  i\geq 3 \\
  \phi^{-1}(BC) \textrm{ if } i=2 \textrm{ and } \phi(w_1w_2) =
  ABC \end{cases}
$$
Then for any $i \geq 2$, $q(w_{i-1},w_i)$ is a numeral indicating the
quadrant in which $p_{i}$ lies with respect to $\{p_0, \ldots,
p_{i-2}\}$.
\end{lem}
\begin{proof}
  It is obvious that $q(w_{i-1},w_i)$ is a numeral. The fact that it
  indicates the claimed quadrant is proved by case examination,
  distinguishing the case $i\geq 3$ from $i=2$.

If $i \geq 3$, $w_{i-1}$ and $w_{i}$ are directions. For example if
$w_{i-1}=L$ and $w_{i}=U$, then $p_{i}$ lies in the quadrant 2 and
$\phi^{-1}(LU)=2$.
 
If $i=2$, $w_{i-1}$ is a numeral and $w_{i}$ is a direction. For
example if $w_{i-1}=1$ and $w_{i}=L$, then $p_{i}$ lies in the
quadrant~2 and we have $\phi(1L)=RUL$ and $\phi^{-1}(UL)=2$.
\end{proof}

By Remarks \ref{rem:simplepin} and \ref{rem:proper_strict} a., pin
words encoding simple permutations are either strict or quasi
strict. We first show how to interpret $\preceq$ by a factor relation
in the case of strict pin words.

\begin{lem}
\label{prop phi}
For any strict pin words $u$ and $w$, $u \preceq w$ if and only if $
\phi(u)$ is a factor of $ \phi(w)$.
\end{lem}
\begin{proof}
  If $u \preceq w$, as $u$ is a strict pin word, writing $u$ in terms
  of its strong numeral-led factors leads to $u = u^{(1)}$, thus $w$
  can be decomposed into a sequence of factors $w =
  v^{(1)}w^{(1)}v^{(2)}$ as in  Definition
  \ref{def:preceq}.

  If $v^{(1)}$ is empty then $w^{(1)}$ begins with a numeral, $w^{(1)}
  = u^{(1)}$ and $u$ is a prefix of $w$. Consequently $ \phi(u)$ is a
  prefix of $ \phi(w)$.

  Otherwise $i=|v^{(1)}| \geq 1$ and $w^{(1)}$ begins with a
  direction. By Definition \ref{def:preceq}, the first letter
  $w_{i+1}$ of $w^{(1)}$ corresponds to a point $p_{i+1}$ lying in the
  quadrant specified by $u_{1}$ (the first letter of $u^{(1)}$), and
  all other letters (which are directions) in $u^{(1)}$ and $w^{(1)}$
  agree: $u_{2} \ldots u_{|u|} = w_{i+2} \ldots w_{i+|u|}$.

  By Lemma \ref{lem:quadrant}, $q(w_i,w_{i+1})$ is the quadrant in
  which $p_{i+1}$ lies, \emph{i.e.} $u_1 = q(w_i,w_{i+1})$. Since $|u|
  \geq 2$, by definition of $q$ we have that $\phi(u) =
  \phi(q(w_i,w_{i+1})w_{i+2} \ldots w_{i+|u|})$ is a factor of
  $\phi(w)$.

\medskip

Conversely if $ \phi(u)$ is a factor of $ \phi(w)$ then $
\phi(w)=v\, \phi(u)\, v'$.  If $v$ is empty then $ \phi(u)$ is a prefix
of $ \phi(w)$ thus $u$ is a prefix of $w$ hence $u \preceq w$.

If $|v|=i \geq 1$ then by definition of $\phi$, $u_2 \ldots u_{|u|}$
is a factor of $\phi(w)$ and more precisely appears in $\phi(w)$ for
indices from $i+3$ to $i+|u|+1$. This means that $u_2 \ldots u_{|u|} =
w_{i+2} \ldots w_{i+|u|}$. Since $i\geq 1$, $w_{i+1}$ is a direction,
and we are left to prove that the point $p_{i+1}$ corresponding to
$w_{i+1}$ lies in the quadrant indicated by $u_1$. By Lemma
\ref{lem:quadrant}, $p_{i+1}$ lies in quadrant $q(w_i,w_{i+1})$, and
we easily check that $q(w_i,w_{i+1}) = \phi^{-1}(xy)$ where $xy$ are
the first two letters of $\phi(u)$. Hence, we get that $q(w_i,w_{i+1})
=u_1$, concluding the proof.
\end{proof}

The second possible structure for a pin word corresponding to a simple
permutation is to begin with two numerals.

\begin{lem}
\label{prop semi}
Let $u$ be a quasi strict pin word and $w$ be a strict pin word.  
If  $u\preceq w$ then  $\phi(u_{2} \ldots u_{|u|})$ is
a factor of $\phi(w)$ which begins at position $p \geq 3$.
\end{lem}

\begin{proof}
  Decompose $u$ into its strong numeral-led factors $u =
  u^{(1)}u^{(2)}$. Notice that $u^{(2)} = u_{2} \ldots u_{|u|}$.
  Since $u\preceq w$, $w$ can be decomposed into a sequence of factors
  $w = v^{(1)}w^{(1)}v^{(2)}w^{(2)}v^{(3)}$ satisfying Definition
  \ref{def:preceq}. Moreover $|w^{(1)}|=|u^{(1)}|=1$ so $w^{(2)}$
  contains no numerals thus $v^{(2)}$ is non-empty, the first letter
  of $w^{(2)}$ corresponds to a point lying in the quadrant specified
  by the first letter of $u^{(2)}$, and all other letters in $u^{(2)}$
  and $w^{(2)}$ agree. Hence $w = v^{(1)}w^{(1)}v \phi
  (u^{(2)})v^{(3)}$ where $v$ is the prefix of $v^{(2)}$ of length
  $|v^{(2)}|-1$. Then $\phi (u^{(2)})$ is a factor of $w$ which has no
  numeral thus $\phi (u^{(2)})$ is a factor of $\phi (w)$ which begin
  at position $p \geq 3 $.
\end{proof}

\begin{lem}
 \label{prop semi2}
Let $u$ be a quasi strict pin word corresponding to a permutation
$\pi$ and $w$ be a strict pin word corresponding to a permutation
$\sigma$.  If $\phi(u_{2} \ldots u_{|u|})$ is
a factor of $\phi(w)$ which begins at position $p \geq 3$ then $\pi$ is
a pattern of $\sigma$.
\end{lem}

\begin{proof}
Set $u^{(2)}=u_{2} \ldots u_{|u|}$. 
Since $\phi(u^{(2)})$ is a factor of $\phi(w)$ which begins at
position $p \geq 3$ then by Lemma \ref{prop phi}, $u^{(2)}\preceq
w$. Let $p_{1} \dots p_{n}$ be a pin representation of $w$ (which
corresponds to $\sigma$) and $\Gamma$ be the subset of points
corresponding to $u^{(2)}$, then $\Gamma \subseteq \{p_{3} \dots
p_{n}\}$. Let $\pi '$ be the permutation corresponding to $\{p_{1}\}
\cup \Gamma$, then $\pi' \leq \sigma$. We claim that $\pi'=\pi$.  Let
$i$ be the quadrant in which $p_{1}$ lies, and $v = i\,u^{(2)}$. Then
$v$ is a pin word corresponding to $\pi '$.
As $u$ begins with two numerals, there is $k \in \{1, \ldots, 4\}$
such that $u=k\,u^{(2)}$. It is easy to see that $v$ 
and $u$ encode the same permutation, even if $i\neq k$. Hence $\pi'=\pi$. 
\end{proof}

The set of pin words of any simple permutation $\pi$ contains at most
$64$ elements. Indeed by Lemma 4.6 of \cite{BBR09} there are at most
$8$ pin representations $p$ of $\pi$ --corresponding to the possible
choices of $(p_1, p_2)$-- and at most $8$ pin words for each pin
representation (see Figure~\ref{fig:origine}) so at most 64 pin words for $\pi$. 

We define $E(\pi) = \{
\phi(u)\,| u$ is a strict pin word corresponding to $\pi\} \cup \{ v
\in M\,| $ there is a quasi strict pin word $u$ corresponding to $\pi$
and $x\in \{LU,LD,RU,RD\}\uplus \{UL,UR,DL,DR \}$ such that $v=x\,
\phi(u_{2} \ldots u_{|u|}) \}$.  For the second set, the first letter
of $\phi(u_{2} \ldots u_{|u|})$ determines the set in which $x$ lies.

By Remarks \ref{rem:simplepin} and \ref{rem:proper_strict} a., the pin
words of $\pi$ are either strict or quasi strict, therefore $|E(\pi)|
\leq 64\times 4 = 256$.

\begin{theo}
\label{prop factor}
Let $\pi$ be a simple permutation and $w$ be a strict pin word
corresponding to a permutation $\sigma$. Then $\pi \nleq \sigma$ if
and only if $\phi(w)$ avoids the finite set of factors $E(\pi)$.
\end{theo}

Notice that it is enough to consider only one strict pin word
corresponding to $\sigma$ rather than all of them.

\begin{proof}
  If $\pi \leq \sigma$, then by Lemma \ref{csq ordre}, there is a pin
  word $u$ corresponding to $\pi$ with $u \preceq w$. By Remarks
  \ref{rem:simplepin} and \ref{rem:proper_strict} a., $u$ is a strict pin
  word or a quasi strict pin word. If $u$ is a strict pin word then,
  by Lemma \ref{prop phi}, $\phi(u)$ is a factor of $\phi(w)$ so
  $\phi(w)$ has a factor in $E(\pi)$. If $u$ is a quasi strict pin
  word then by Lemma \ref{prop semi}, $\phi(u_{2} \ldots u_{|u|})$ is
  a factor of $\phi(w)$ which begins at position $p \geq 3$. Let $x$
  be the two letters preceding $\phi(u_{2} \ldots u_{|u|})$ in
  $\phi(w)$. As $\phi(w) \in M$, $x\phi(u_{2} \ldots u_{|u|})$ is a
  factor of $\phi(w)$ that belongs to $E(\pi)$.

 Conversely suppose that $\phi(w)$ has a factor $v$ in $E(\pi)$. If
$v\in \{ \phi(u)\,| u$ is a strict pin word corresponding to $\pi\}$
then by Lemma \ref{prop phi}, there is a pin word $u$ corresponding to
$\pi$ with $u \preceq w$ so by Lemma \ref{csq ordre}, $\pi \leq
\sigma$. Otherwise there is a quasi strict pin word $u$ corresponding
to $\pi$ and $x\in \{LU,LD,RU,RD,UL,UR,DL,DR \}$ such that $v=x\,
\phi(u_{2} \ldots u_{|u|})$ is a factor of $\phi(w)$. Thus $\phi(u_{2} \ldots u_{|u|})$
is a factor of $\phi(w)$ which begins at position $p \geq 3$ and by
Lemma \ref{prop semi2}, $\pi \leq \sigma$.
\end{proof}

Returning to our motivation with respect to the number of proper
pin-permutations in $Av(B)$, the links between pattern containment
relation and pin words that we established yield Theorem~\ref{lem:wcclass}. 

\begin{theo}\label{lem:wcclass}
A wreath-closed class $Av(B)$ has arbitrarily long proper
pin-permutations if and only if there exist words of arbitrary length
on the alphabet $\{ L,R,U,D\}$ avoiding the set of factors $\cup_{\pi
\in B} E(\pi) \cup \{ LL,LR,RR,RL,UU,UD,DD,DU\}$.
\end{theo}

\begin{proof}
The class $Av(B)$ contains arbitrarily long proper pin-permutations if
and only if there exist arbitrarily long proper pin-permutations which
have no pattern in $B$. That is --making use of Theorem \ref{prop
factor} and Remark \ref{rem:proper_strict} b.--, if and only if there
exist arbitrarily long strict pin words $w$ such that $\phi(w)$ avoids
the set of factors $\cup_{\pi \in B} E(\pi)$, or equivalently if and
only if there exist words of arbitrary length on the alphabet $\{
L,R,U,D\}$ which avoid the set of factors $\cup_{\pi \in B} E(\pi)
\cup \{ LL,LR,RR,RL,UU,UD,DD,DU\}$.
\end{proof}

\section{From the finiteness problem to a co-finiteness problem}
\label{sec:complexity} 

We are now able to give the general algorithm to decide if a
wreath-closed permutation class given by its finite basis $B$ contains a
finite number of proper pin-permutations (see Algorithm
\ref{alg:polynomial}). 

In this algorithm, ${\mathcal P}_B$ denotes the set of pin words that
encode the permutations of $B$, ${\mathcal L}({\mathcal P}_B)$ the
language of words on the alphabet $\{ L,R,U,D \}$ which contain as a
factor a word of $\cup_{\pi \in B} E(\pi)$ or one of the $8$ words
$LL,LR,RR,RL,UU,UD,DD$ and $DU$, $\mathcal A$ the automaton
recognizing ${\mathcal L}({\mathcal P}_B)$ and ${\mathcal A^c}$ the
automaton that recognizes the complementary language of ${\mathcal
  L}({\mathcal P}_B)$ in $\{L,R,U,D\}^{*}$. Notice that ${\mathcal
  A^c}$ recognizes the words $\phi(w)$ for $w$ a strict pin word
encoding a proper pin-permutation $\sigma \in Av(B)$.

\smallskip

\begin{algorithm}[H]
\SetKwData{B}{B}\SetKwData{PB}{{\mathcal P}_B}\SetKwData{Up}{up}
\SetKwFunction{PinWords}{{\sc PinWords}}\SetKwFunction{FindCompress}{FindCompress}
\SetKwInOut{Input}{{input}}\SetKwInOut{Output}{output}
\BlankLine
\Input{a set $B$ of simple permutations}
\Output{boolean : true if and only if $Av(B)$ contains only a finite number of proper pin-permutations}
\BlankLine
${\mathcal P}_B\leftarrow$\PinWords{$B$} \tcp{ Determine the set of pin words associated to the elements of $B$}
${\mathcal A}\leftarrow$ \textsc{Automaton}\texttt{(}${\mathcal
  L(\mathcal P}_B)$\texttt{)} \tcp{Build a complete deterministic automaton recognizing ${\mathcal L({\mathcal P}_B)}$}
\eIf{${\mathcal A^c}$ contains an accessible and co-accessible cycle}{\Return false}{\Return true} 

\caption{Deciding the finiteness of the number of proper
  pin-permutations }\label{alg:polynomial}
\end{algorithm}

\smallskip

The first part of this algorithm relies on the function {\sc Pinwords}
(described by Algorithm \ref{alg:pinword}) which computes the pin
words associated to a simple permutation.  It uses the fact that the
pin representations of a simple permutation, when they exist, are
always proper (see Remark \ref{rem:simplepin}), and that from Lemma
4.3 in \cite{BBR09}, the first two pins of a proper pin representation are
in {\em knight position} ({\it i.e.}, in a configuration like
\begin{tikzpicture}[scale=.15] \draw [help lines] (0,0) grid (3,2);
\fill (0.5,1.5) circle (6pt); \fill (2.5,0.5) circle (6pt);
\end{tikzpicture} or one of its $3$ symmetries under rotation and 
reflection). Next from two points in knight position, a proper pin 
representation, if it exists, can be efficiently computed using the
separation condition. Finally it remains to encode the pin representation by
pin words.

\begin{algorithm}
\SetKwData{B}{B}\SetKwData{Count}{Count}
\SetKwInOut{Input}{input}\SetKwInOut{Output}{output}
\BlankLine
\Input{a simple permutation $\sigma$}
\Output{The set $P$ of pin words encoding $\sigma$}
\BlankLine
\tcp{Count the number of ordered pairs of points in knight position}
 $E \leftarrow \varnothing$\; 
\ForEach{$\sigma_i$}{$E \leftarrow E \bigcup \{ (\sigma_i,\sigma_j)$ in knight position$\}$}
\tcp{If more than $48$ pairs are found, $\sigma$ is not a pin-permutation}
\If{$|E|>48$}{\Return $\varnothing$}
\tcp{Otherwise each knight may be the beginning of a pin representation
  of $\sigma$}
$P \leftarrow \varnothing$\;
\ForEach{$(\sigma_i,\sigma_j) \in E$}{$P \leftarrow P \bigcup \{$ 
pin words of the pin representation beginning with  $(\sigma_i,\sigma_j) \}$}
\Return{$P$}

\caption{{\sc Pinwords} function}\label{alg:pinword}
\end{algorithm}

\begin{lem}\label{lem:algo pwd}
 Algorithm \ref{alg:pinword} computes the set of pin words encoding a
 simple permutation $\sigma$ in linear time with respect to the length $n$ of $\sigma$.
\end{lem}

\begin{proof}
Algorithm \ref{alg:pinword} can be decomposed into two parts. First, we count
the number of ordered pairs of points in knight position that should be smaller than
$48$.  Indeed from Lemma~4.4 of \cite{BBR09}, if $\sigma$ is a simple
pin-permutation of length $n$, in any of its pin representations $(p_1,\ldots,p_n)$,
every unordered pair of points $\{p_i,p_j\}$ that is a knight contains at least one of the points $p_1,p_2$
or $p_n$. As only $8$ points can be in knight position with a given
point, the permutation $\sigma$ has at most $24$ unordered pairs of points in knight position, hence at most $48$ ordered pairs $(p_i,p_j)$ that are knights.

Therefore given a simple permutation $\sigma$, we count the number of
ordered pairs of points in knight position. 
To do this, we take each
point $p$ of the permutation and we check if another point is in
knight position with $p$. As at most $8$ cells can contain a point in
knight position with $p$, this counting part runs in time $8
n$. 

If this number is greater than $48$, $\sigma$ is not a
pin-permutation. Otherwise, the second part of the algorithm computes,
for each ordered pair of points in knight position, the pin representation
beginning with it (if it exists) and its associated pin words. This
can also be done in linear time as there is at most one pin
representation of $\sigma$ beginning with a given ordered pair of points.
Indeed, because $\sigma$ is simple, its pin representations are always
proper (see Remark \ref{rem:simplepin}).  The pin representation
starting with a given knight is then obtained as follows. If
$(p_1,\ldots, p_i)$ has already been computed then, since the pin
representation we look for is proper, $p_{i+1}$ separates $p_i =
\sigma_k$ from previous points. It means that either it separates them
vertically, and then $p_{i+1}= \sigma_{k+1}$ or
$p_{i+1}=\sigma_{k-1}$, or it separates them horizontally and then its
value must be $\sigma_k \pm 1$. Therefore, if we compute $\sigma^{-1}$
in advance (which is easily done by a linear-time precomputation) we
are allowed to find the next point in a proper pin representation in
constant time.

Finally as at most $8$ pin words (choice of the origin, see Figure
\ref{fig:origine}) correspond to a given pin representation, computing
all pin words can easily be done in linear time from the pin
representation.
\end{proof}

\begin{lem}\label{lem: main algo} 
Algorithm \ref{alg:polynomial} tests if a wreath-closed
permutation class given by its finite basis $B$ contains a finite number
of proper pin-permutations in linear time with respect to $n = \sum_{\pi \in B} |\pi|$.
\end{lem}

\begin{proof}
First according to Lemma~\ref{lem:algo pwd} the {\sc Pinwords}
function applied to the $|B|$ patterns of the basis runs in total time
$\mathcal{O}(n)$, and produces a set ${\mathcal P}_B$, containing at
most $|B|\cdot 48 \cdot 8$ words, whose lengths sum to
$\mathcal{O}(n)$.

Next a complete deterministic automaton ${\mathcal A}$ recognizing
${\mathcal L}({\mathcal P}_B)$ the set of words having a factor in
$\cup_{\pi \in B} E(\pi) \cup \{ LL,LR,RR,RL,UU,UD,DD,DU\}$ can be
built in linear time (w.r.t. $n$) using Aho-Corasick algorithm
\cite{AC75}. With this construction the number of states of the
resulting automaton is also linear. The automaton ${\mathcal A^c}$
that recognizes the complementary language of ${\mathcal L}({\mathcal
  P}_B)$ in $\{L,R,U,D\}^{*}$ is obtained by exchanging final and
non-final states of the initial automaton ${\mathcal A}$ which is
complete and deterministic.  Then it remains to test in the complete
deterministic automaton ${\mathcal A^c}$ whether there exists an
accessible cycle from which a path leads to a final state ({\em i.e.},
that is co-accessible). Making use of a depth-first traversal, this
step takes a linear time.  Hence checking if there exist arbitrarily
long words on $\{ L,R,U,D\}$ which avoid a finite set of factors can
be done in linear time -- linear in the sum of the lengths of the
factors. Together with Theorem~\ref{lem:wcclass} this concludes the proof.
\end{proof}

The preceding results allow us to decide in linear time if a
wreath-closed permutation class given by its finite basis contains
arbitrarily long proper pin-permutations. To end our proof, following
the same steps as \cite{BRV06}, we must deal with wedge simple
permutations and parallel alternations in order to decide if the
permutation class contains a finite number of simple
permutations. These results are summarized in the following theorem:

\begin{theo}\label{thm:thm}
 Let $Av(B)$ be a finitely based wreath-closed class of
 permutations. Then there exists an algorithm to decide in time
 ${\mathcal O}(n \log n)$ where $n=\sum_{\pi \in B} |\pi|$ whether
 this class contains finitely many simple permutations.
\end{theo}

\begin{proof}
  From Theorem \ref{thm:brignall}, we can look separately at parallel
  alternations, wedge simple permutations and proper
  pin-permutations. For parallel alternations and wedge simple
  permutations, Lemma \ref{lem:complexity} shows that testing if their
  number in $Av(B)$ is finite can be done in ${\mathcal O}(n \log n)$
  time. The case of proper pin-permutations can be solved with
  Algorithm~\ref{alg:polynomial}. From Lemma~\ref{lem: main algo}
  checking if there exist arbitrarily long proper pin-permutations in
  a wreath-closed permutation class can done in linear time -- linear
  in the sum of the lengths of the elements of the basis of the class
  -- concluding the proof.
\end{proof}

\paragraph*{Conjecture} We strongly believe that Theorem \ref{thm:thm} has a 
generalization to all finitely based permutation classes (and not only
wreath-closed classes): namely, we expect that the complexity of
deciding whether a finitely based permutation class contains a finite
number of simple permutations is polynomial, however of higher
degree. 

\medskip

This complexity gap we foresee for this generalization and the growing
importance of wreath-closed classes in the permutation patterns field
justify to our eyes the interest of the result proved in this article.
Furthermore, in the general case the pattern relation on permutations
cannot be translated into a factor relation on the pin words that
encode these permutations. However the substitution decomposition
of pin-permutations described in \cite{BBR09} should allow us to obtain
an efficient recursive algorithm.

\paragraph*{Open problem} 
By \cite{AA05}, containing a finite number of simple permutations is a
sufficient condition for a permutation class to have an algebraic
generating function. Our work allows to decide efficiently whether the
number of simple permutations in the class is finite, but does not
allow the computation of the \emph{set} of simple permutations in the
class. Describing an efficient (polynomial?) procedure solving this
question, and thereafter being able to compute algorithmically the
algebraic generating function associated to the class, would be
natural continuations of our work.

\paragraph*{Acknowledgement}
The authors wish to thank the referees for their helpful comments and
suggestions, that brought our work to a clarified presentation.

\bibliographystyle{plain} \bibliography{biblio}

\begin{thebibliography}{10}

\bibitem{AC75}
Alfred~V. Aho and Margaret~J. Corasick.
\newblock Efficient string matching: An aid to bibliographic search.
\newblock {\em Communications of the ACM}, 18(6), June 1975.

\bibitem{AAAH01}
Michael~H. Albert, Robert E.~L. Aldred, Mike~D. Atkinson, and Derek~A. Holton.
\newblock Algorithms for pattern involvement in permutations.
\newblock In {\em ISAAC '01: Proceedings of the 12th International Symposium on
  Algorithms and Computation}, volume 2223 of {\em Lecture Notes in Computer
  Science}, pages 355--366, London, UK, 2001. Springer-Verlag.

\bibitem{AA05}
Michael~H. Albert and Mike~D. Atkinson.
\newblock Simple permutations and pattern restricted permutations.
\newblock {\em Discrete Mathematics}, 300(1-3):1--15, 2005.

\bibitem{ALR05}
Michael~H. Albert, Steve Linton, and Nik Ru{\v{s}}kuc.
\newblock The insertion encoding of permutations.
\newblock {\em Electron. J. Combin.}, 12:Research Paper 47, 31 pp.
  (electronic), 2005.

\bibitem{ARS09}
Mike~D. Atkinson, Nik Ru\v{s}kuc, and Rebecca Smith.
\newblock Substitution-closed pattern classes.
\newblock To appear in {\em J. Combin. Theory Ser. A}.

\bibitem{BBR09}
Frédérique Bassino, Mathilde Bouvel, and Dominique Rossin.
\newblock Enumeration of pin-permutations.
\newblock Technical report, Université Paris Diderot, Université Paris Nord,
  2009.

\bibitem{BHV06b}
Robert Brignall, Sophie Huczynska, and Vincent Vatter.
\newblock Decomposing simple permutations, with enumerative consequences.
\newblock {\em Combinatorica}, 28(4):385--400, jul 2008.

\bibitem{BHV06a}
Robert Brignall, Sophie Huczynska, and Vincent Vatter.
\newblock Simple permutations and algebraic generating functions.
\newblock {\em J. Combin. Theory Ser. A}, 115(3):423--441, 2008.

\bibitem{BRV06}
Robert Brignall, Nik Ru{\v{s}}kuc, and Vincent Vatter.
\newblock Simple permutations: decidability and unavoidable substructures.
\newblock {\em Theoret. Comput. Sci.}, 391(1-2):150--163, 2008.

\bibitem{KiMa03}
Sergey Kitaev and Toufik Mansour.
\newblock A survey on certain pattern problems.
\newblock Preprint available at
  http://www.ru.is/kennarar/sergey/publications.html, 2003.

\bibitem{Knuth:ArtComputerProgramming:1:1973}
Donald~E. Knuth.
\newblock {\em Fundamental Algorithms}, volume~1 of {\em The Art of Computer
  Programming}.
\newblock Addison-Wesley, Reading MA, 3rd edition, 1973.

\bibitem{MaTa04}
Adam Marcus and G{\'a}bor Tardos.
\newblock Excluded permutation matrices and the {S}tanley-{W}ilf conjecture.
\newblock {\em J. Comb. Theory, Ser. A}, 107(1):153--160, 2004.

\bibitem{Vat05}
Vincent Vatter.
\newblock Enumeration schemes for restricted permutations.
\newblock {\em Comb. Probab. Comput.}, 17(1):137--159, 2008.

\end{thebibliography}
\end{document}